\documentclass{amsart}

\usepackage{amssymb}
\usepackage{amsthm}
\usepackage{amsmath,amsfonts,amsthm}
\usepackage{thmtools}
\usepackage{etoolbox}
\usepackage{mathtools}
\usepackage{soul}
\usepackage{tcolorbox}
\usepackage{times}
\usepackage{algorithm}
\usepackage{algorithmicx}
\usepackage{algpseudocode}
\usepackage{tabularx}
\usepackage{caption}
\usepackage{latexsym} 
\usepackage{fleqn}
\usepackage{tikz}
\usepackage{subfiles}
\usepackage{microtype}
\usepackage{url}

\newtheorem{definition}{Definition}

\newtheorem{lemma}[definition]{Lemma}
\newtheorem{theorem}[definition]{Theorem}

\newtheorem{example}[definition]{Example}

\makeatletter
\newcommand{\multiline}[1]{%
  \begin{tabularx}{\dimexpr\linewidth-\ALG@thistlm}[t]{@{}X@{}}
    #1
  \end{tabularx}
}
\makeatother

\newcommand{\allen}{\ensuremath{\mathcal A}}

\newcommand{\problemDef}[3]
{%
    
    \begin{tcolorbox}[arc=0.1mm,boxsep=-0.6mm,left=1.9mm,right=1.9mm,bottom=1.4mm,top=1.4mm,adjusted title={\strut \sc#1},colback=white!5]

    \noindent\textbf{Instance:} #2
    
    \noindent\textbf{Question:} #3
    \end{tcolorbox}
}

\newcommand{\Ordo}{{\mathcal{O}}}

\newcommand{\stub}{{stub}}

\newcommand{\pot}{\textsc{POT}}
\newcommand{\toop}{TOP}

\newcommand{\ptoop}{PTOP}

\newcommand{\rpartial}{\ensuremath{R_{\mathrm{par}}}}
\newcommand{\rtotal}{\ensuremath{R_{\mathrm{tot}}}}
\newcommand{\rcorrect}{\ensuremath{R_{\mathrm{corr}}}}

\usetikzlibrary{arrows,positioning,automata}

\tikzset{
>=stealth', 
node distance=2.4cm, 
every state/.style={thick, fill=gray!10}, 
initial text=$ $, 
}

\author{Leif Eriksson}
\address[L. Eriksson]%
   {Dep. Computer and Information
     Science, \\Link\"opings  Universitet, Sweden
   }
\email{leif.eriksson@liu.se}   
\author{Victor Lagerkvist}
\address[V. Lagerkvist]%
   {Dep. Computer and Information
     Science, \\Link\"opings  Universitet, Sweden}
\email{victor.lagerkvist@liu.se}      

\title{A Fast Algorithm for Consistency Checking Partially Ordered Time}

\begin{document}
\maketitle

\begin{abstract}
      Partially ordered  models of time occur naturally in applications  where  agents or processes cannot perfectly communicate with each other, and can be traced back to the seminal work of Lamport. In this paper we consider the problem of {\em deciding} if a (likely incomplete) description of a system of events is consistent, the {\em network consistency problem} for the {\em point algebra of partially ordered time} ({\pot}). While the classical complexity of this problem has been fully settled, comparably little is known of the {\em fine-grained} complexity of {\pot} except that it can be solved in $\Ordo^*((0.368n)^n)$ time by enumerating {\em ordered partitions}. We construct a {\em much} faster algorithm with a  run-time bounded by $\Ordo^*((0.26n)^n)$.
    This is achieved by a sophisticated enumeration of structures similar to total orders, which are then greedily expanded towards a solution.
    While similar ideas have been explored earlier for related problems it turns out that the analysis for {\pot} is non-trivial and requires significant new ideas.
\end{abstract}

\section{Introduction}\label{sec:intro}
{\em Qualitative reasoning} is an important formalism in artificial intelligence where the objective is to reason about continuous properties given certain relations between the unknown entities.
Two important subfields are {\em temporal} reasoning, e.g., the {\em point algebra for partially ordered time} ({\pot}), {\em Allen’s interval algebra} ({\allen}), and the {\em point algebra for branching time}, and {\em spatial} reasoning, e.g., the {\em region connection calculus} (RCC), the {\em cardinal direction calculus}, and the {\em rectangle algebra}.
There are numerous applications of all of these formalisms in AI, e.g., in knowledge representation~\cite{forbus2019qualitative}, linguistics~\cite{DBLP:journals/ai/Allen84,song}, and planning~\cite{DBLP:conf/kr/Allen91,allen2014reasoning,DBLP:conf/sbia/NogueiraFA96,DBLP:journals/ci/SongC96}. For a broad overview with further applications and references we refer to the survey by Dylla et al.~\cite{Dylla:2017:SQS:3058791.3038927}. 

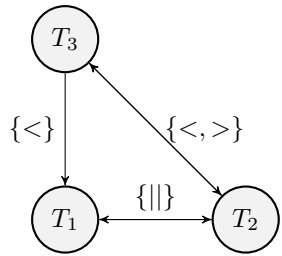
\begin{figure}
\centering
\begin{tikzpicture}
\node[state] (t1) {$T_1$};

\node[state, right of=t1] (t2) {$T_2$};
\node[state, above of=t1] (t3) {$T_3$};

\draw[->] (t3) -- (t1) node[midway, left]{$\{<\}$};
\draw[->] (t2) -- (t1) node[midway, above]{$\{||\}$};
\draw[->] (t1) -- (t2);
\draw[->] (t2) -- (t3) node[midway, right]{$\{<,>\}$};
\draw[->] (t3) -- (t2);

\end{tikzpicture}
\caption{A scenario with tree tasks $T_1, T_2, T_3$ where $T_1$ precedes $T_3$, $T_1$ and $T_2$ are incomparable, and where either $T_2$ precedes $T_3$, or $T_3$ precedes $T_2$.}
\label{fig:pot}
\end{figure}

In this paper we are interested in constructing fast (but superpolynomial) algorithms for NP-hard temporal reasoning problems, with a particular focus on the {\pot} problem. Here, the basic task is to check whether a given set of events and a set of possible relationships between them is {\em consistent} in the sense that there exists some partial ordering of the events which does not contradict any of the given relationships. Thus, this model of time is suitable in applications where we are working with agents/processes who cannot perfectly communicate with each other and where a global, totally ordered model of time is not possible, e.g., in distributed or concurrent systems~\cite{anger89,Anger:etal:flairs99,lamport86}. There are many possible formulations of the basic computational problem and we consider the setup where the possible basic relations among two points in time are less than (${<}$), greater than (${>}$), incomparable to (${||}$) and equivalent to (${=}$), i.e., the basic operations of the well-known {\em point algebra}~\cite{DBLP:journals/jacm/LadkinM94}. To make it possible to encode complex relationship between events we follow Broxwall \& Jonsson~\cite{Broxvall:Jonsson:ai2003} and allow disjunctions of the basic relations.
Crucially, disjunctions make it possible to model incomplete relationships between tasks, e.g., we are given two tasks, where one of them started the other but we lack the knowldge of which one was first. See Figure~\ref{fig:pot} for a visualization of a constraint network with three tasks (where $\{<,>\}$ means that $<$ or $>$ is true). This network is satisfiable since the tasks can be ordered as $T_1 < T_3$, $T_2 < T_3$ where $T_1$ and $T_2$ are incomparable.
Let us also remark that all of the aforementioned problems can be formulated as infinite-domain {\em constraint satisfaction problems} (CSPs) over $\omega$-categorical {\em constraint languages}~\cite{DBLP:journals/jair/BodirskyJ17}. In this framework one first fixes a set of binary {\em basic relations} $\mathcal{B}$ and then consider CSP$(\mathcal{B}^{\vee =})$ where $\mathcal{B}^{\vee =} = \bigcup_{R \in \mathcal{B}} R$  is the union of the basic relations.

For example, the {\pot}  problem can then be formulated as a CSP problem where the basic relations are formed over the {\em random partial order} (cf. Chapter 2 in Bodirsky~\cite{Bodirsky:InfDom}).

\subsection*{Related work}
Significant attention has been devoted to finding maximally tractable subclasses of these problems, typically accomplished by local consistency methods, and for e.g.\ RCC-5, RCC-8, {\allen}, and  {\pot}, all maximal, tractable classes have been identified~\cite{Dylla:2017:SQS:3058791.3038927}.
If one extends to arbitrary first-order reducts of the basic relations the satisfiability problem of {\pot} has a complexity dichotomy by Kompatscher \& Pham~\cite{pham2017}, while the corresponding problem for RCC is generally undecidable. 
For many more examples of complexity dichotomies for infinite-domain CSPs, see e.g.\ Bodirsky~\cite{Bodirsky:InfDom}  

But, naturally, we cannot be content with merely understanding the tractable fragments, since their expressive power is too restrictive to be able to model real-world problems. Hence, we need methods for solving NP-hard reasoning tasks as fast as possible.
To expand our understanding of the NP-hard cases we would thus like to (1) construct algorithms faster than exhaustive search, and (2) prove that certain types of speedups are {\em not} possible, subject to stronger complexity theoretical assumptions than {P~$\neq$~NP}.
Complexity questions like these, especially for a precise complexity parameter such as the number of variables, $n$, typically fall under the scope of {\em fine-grained complexity}.
Thus, given a reasonable set of basic relations ${\mathcal B}$, how fast can we expect to solve CSP$({\mathcal B}^{\vee =})$ (which might be NP-hard even if CSP$({\mathcal B})$ is tractable)?
Here, let us first remark that any CSP$({\mathcal B}^{\vee = })$ problem is solvable by an exhaustive backtracking algorithm in $2^{\Ordo(n^2)}$ time, under mild assumptions on the set of basic relations ${\mathcal B}$ (e.g., that CSP$({\mathcal B})$ is solvable in polynomial time).
However, for several prominent problems in qualitative reasoning, including {\pot}, it is possible to argue that the $2^{\Ordo(n^2)}$ bound is too naive to be used as a baseline for improvement.
Instead, these problems can be solved in $2^{\Ordo(n \log n)}$  time by enumerating {\em ordered partitions}~\cite{lagerkvist2017d,jonsson2021}, pushing down the running time to $\Ordo^*((2n)^{2n})$ for {\allen}, $\Ordo^*((0.531n)^n)$ for RCC-8 and $\Ordo^*((0.368n)^n)$ for RCC-5 and {\pot}\footnote{The notation $\Ordo^*(\cdot)$ suppresses polynomial factors.}.

Thus, these problems can be solved by enumerating objects similar to assignments in finite-domain CSPs, and the question is then whether it is possible to solve the problem faster than exhaustively enumerating all orderings, similar to how it is a major open question whether CNF-SAT is solvable in $\Ordo^*(c^n)$ time for {\em some} $c < 2$.
This is indeed known to be possible for certain reasoning problems, e.g., {\allen}, which recently has been solved in $\Ordo^*((1.0615n)^{n})$ time~\cite{eriksson2021}, and if the problem is restricted to intervals of length one then it can even be solved in  $2^{\Ordo(n \log \log n)}$ time~\cite{DBLP:conf/kr/DabrowskiJOO20}, and if no point occurs inside more than $k$ intervals then it can be  solved in $\Ordo^*(k^n)$ time~\cite{ijcai2022p251}.
A faster $f(k)^n$ time algorithm, for some function $f$, is also known for the special case of {\pot} where a solution with \emph{effective width} of at most $k$ is asked for~\cite{ijcai2022p251}. However, despite these improvements, we are still far away from an unconditional single-exponential $O^*(c^n)$ time algorithm and even further away from the best-known lower bounds which only rule out subexponential algorithms running in $2^{o(n)}$ time under the {\em exponential-time hypothesis}~\cite{jonsson2021}.
Hence, cutting-edge research suggests that qualitative reasoning problems in many cases admit {\em significantly} improved algorithms even though general single-exponential running times seem to be out of reach with existing methods. 

\subsection*{Our contribution}
In this paper we advance this frontier by describing a novel and significantly improved algorithm for the {\pot} problem with a running time of $\Ordo^*((0.26n)^{n})$, which is
{\em much} faster than the previously known baseline of $O^*{((0.368n)^n)}$.
Hence, our algorithm is not only a showcase that an improved algorithm is possible for {\pot} but {\em significantly} beats the naive upper bound based on enumerating ordered partitions.
We achieve this as follows: after 
introducing the necessary technical background (in Section~\ref{sec:prelim}) we start our work on {\pot} in Section~\ref{sec:partial}.
We analyze the structures of potential solutions and use greedy choices to find a structure that is suitable for enumeration and which yields a significant improvement over enumerating ordered partitions. The basic idea is to group variables into pairs and then order these pairs relative to each other instead of all variables individually, and
ordering the variables in each pair can thereafter be done greedily.
Hence, the basic idea is not that complex, but actually proving soundness and completeness of our approach is non-trivial and requires novel techniques.
Finally, we conclude our results in Section~\ref{sec:conc} and present a discussion over how these algorithms and the ideas behind them might be open to further improvements and what other problems these ideas may be applicable to. Notably, can the algorithm be adapted to solve RCC-5 or RCC-8, and how far can we push the upper bound with this algorithmic technique?

\section{Preliminaries}\label{sec:prelim}

Given a set of finitary relations $\Gamma$ defined on a (potentially infinite) set $D$ of values, we define the constraint satisfaction problem over $\Gamma$ (CSP$(\Gamma)$) as follows.
\problemDef{CSP$(\Gamma)$}
{
A tuple $(V, C)$, where $V$ is a set of variables and $C$ a set of constraints of the form $R(v_1, \ldots, v_t)$, where $t$ is the arity of $R \in \Gamma$ and $v_1, \ldots, v_t \in V$.
}
{
Is there a function $f \colon V \rightarrow D$ such that $(f(v_1), \ldots, f(v_t)) \in R$ for every $R(v_1, \ldots, v_t) \in C$?
}

The set $\Gamma$ is referred to as a \emph{constraint language} and the function $f$ is sometimes called a {\em satisfying assignment} of an instance, $I$, or simply a \emph{model} of $I$.
We write $||I||$ for the number of bits required to represent an instance $I$ of CSP$(\Gamma)$. 

\begin{definition}\label{partialorder}
We define the following orders:
\begin{enumerate}
    \item 
A pair $(S, \leq)$ is a \emph{partial order} if $\leq$ is reflexive ($\forall x \in S$ then $x \leq x$), antisymmetric ($\forall x, y \in S$, if $x \leq y$ and $y \leq x$ then $x = y$), and transitive (if $x\leq y$ and $y\leq z$ then $x\leq z$).
\item
A pair $(S, <)$ is a \emph{strict partial order} if $<$ is irreflexive, asymmetric ($\forall x, y \in S$, if $x < y$ then $y < x$ does not hold), and transitive.
\item 
A pair $(S, \leq)$ is a \emph{total order} if $\leq$ is reflexive, antisymmetric, transitive and strongly connected ($\forall x, y \in S$ then $x\leq y$ or $y\leq x$).
\end{enumerate}
\end{definition}

If $\odot \in \{<,>,||,=\}$ and $P = (S, \leq_P)$ is a partial (or total) order then we write $\odot_P$ for the relation induced by $P$: $x <_P y$ if $x \leq_P y$ and $y \leq_P x$ does not hold, conversely for $>_P$, $||_P$ if neither $x \leq_P y$ nor $y \leq_P x$, and $x =_P y$ if $x \leq_P y$ and $y \leq_P x$. We now define the main problem of the paper.

\problemDef{Partially Ordered Time}
{A set of variables $V$ and a set of binary constraints $C$ where ${c \subseteq \{<,>,||,=\}}$ for each $c(x,y) \in C$.}
{Is there a partial order $P = (S, \leq)$ with ${|S|\leq |V|}$ and a function $f \colon V\rightarrow S$ such that for every constraint $c(x,y)\in C$, $f(x) \odot_P f(y)$ for some $\odot \in c$?}

Alternatively one can also use the {\em random partial order} $P$, i.e., the (unique) countable partial order which is {\em universal} (contains an isomorphic copy of every finite partial order) and is {\em homogeneous} (any isomorphism between finite substructures can be extended to an automorphism of $P$).
Then, the {\pot}  problem can equivalently well be defined as CSP$(\mathcal{RP})$ where $\mathcal{RP}$ is the closure of $\{<_P,>_P,||_P,=_P\}$ under union (cf.~\cite{Bodirsky:InfDom}).

\section{Partially Ordered Time}\label{sec:partial}

 Our approach to beat the naive $O^*{((0.368n)^n)}$ algorithm for {\pot} involves exploring a carefully selected group of partial orders. The algorithm, in particular, organizes variables into pairs where we only have to consider a relative ordering with $\frac{n}{2}$ other variables. This scheme leads to a runtime that is dominated by $n!/2^{\frac{n}{2}}$. Demonstrating the correctness of this strategy is a nontrivial task, and the analysis itself is arguably as interesting as the precise bound we attain.

\subsection{Definitions}
We start by introducing the concepts necessary for the soundness and completeness proofs of the main algorithm.

\begin{definition}
If $P=(S,\leq_P)$ and $P'=(S',\leq_{P'})$ $P$ are two partial orders then $P$ is a \emph{\stub{}} of $P'$ if  $S \subseteq S'$ and $\leq_P\,\subseteq\,\leq_{P'}$.
\end{definition}

We chose to represent our instances as (multi-)relational networks rather than as sets of constraints.
This will give us more flexibility when adding additional restrictions to our instances, since it allows us to limit ourselves to sets of partial orders under some restrictions.

\begin{definition}
For an arbitrary {\pot} instance $I=(V,C)$ we define two different variants of \emph{relational networks}:
\begin{enumerate}
\item
A function $f \colon V^2 \xrightarrow{} \mathbf{\{<,>,||,=\}}$ is a \emph{relational network} (over $V$), if for every constraint  $c(x,y)\in C$ then $f(x,y) \in c(x,y)$.
We also say that $f$ is a relational network for $I$. 

\item A function $f \colon V^2 \xrightarrow{} {\mathcal P}(\{<,>,||,=\})$ is a \emph{multi relational network} if for every constraint $c(x,y)\in C$ then $f(x,y)\subseteq c(x,y)$.
We also say that $f$ is a (multi) relational network for $I$.  
\end{enumerate}

If $I$ is a 'yes'-instance, i.e. there exists a partial order $(S,\leq_P)$ such that for all $c(x,y)\in C$ then $x\odot_Py$ with ${\odot \in c(x,y)}$, we say that $f$ is a \emph{consistent} (multi) relational network.
\end{definition}

We will occasionally view (multi) relational networks as sets of constraints in the obvious way.
For example, if $f$ is a multi relational network of an instance ${I=(V,C)}$, then $f\cup \{x\{<,>\}y\}$ is equivalent to the multi relational network for the instance ${(V,C\cup{\{x\{<,>\}y\}})}$.

\begin{definition} Consider the set of all multi relational networks over a fixed set $V$. 
\begin{enumerate}
    \item 
For two distinct multi relational networks $f$ and $f'$ over $V$ we write $f \preceq f'$ if  $f(x,y)\subseteq f'(x,y)$ for all $x,y\in V$.
\item
We write $f \prec f'$ for the corresponding irreflexive order.
    \item
    A multi relational network $f$ is said to be {\em maximally general} if it is a maximal element in $\preceq$, i.e., there does not exist $f'$ such that $f \prec f'$.
\end{enumerate}
\end{definition}

A solution for a {\pot} instance $I$ with relational network $f$ can now be represented by a relational network $g\preceq f$. 
This will be convenient since it eliminates the need to refer to fixed values in the context of a solution. 

We also need a local consistency definition for our multi relational networks.

\begin{definition}
A multi relational network $f$ over $V$ is (locally) consistent over $s \subseteq V$ if $f(x,y)=\bigcup g_i(x,y)$, $x,y\in s$ where $\{g_1,\dots,g_n\}$ is the set of all consistent relational networks $g_i\preceq f$ over $s$.
\end{definition}

Next, we define the central concept of composing a partial order with a multi relational network, roughly meaning that the partial order is used to simplify the multi relational network as much as possible.

\begin{definition} \label{def:FunderP}
 Given a {\pot} instance $I=(V,C)$ with multi relational network $f$ and a partial order $P=(V,\leq_P)$, we define the {\pot} instance 
 \begin{equation*}
    P\circ f =
    \begin{cases}
        f(x,y)\setminus\{>\}, &\text{ if } x<_Py, \\
        f(x,y)\setminus\{<\}, &\text{ if } y<_Px, \\
        f(x,y)\cap\{=\}, &\text{ if } y=_Px, \\
        f(x,y), &\text{ if } y||_Px.
    \end{cases}
\end{equation*}
\end{definition}

Note that $P\circ f$ may have cases where $f(x,y)=\emptyset$, meaning that there is no partial order $P'=(V,\leq_{P'})$ satisfying $I$ for which $P$ is a \stub{}.

\begin{lemma}\label{lem:PfyesIyes}
Let $I=(V,C)$ be a {\pot} instance with multi relational network $f$. If there exists a partial order ${P=(V,\leq_P)}$ such that $P\circ f$ is a 'yes'-instance then $I$ is a 'yes'-instance.
\end{lemma}
\begin{proof}
By Definition~\ref{def:FunderP}, $(P\circ f)\preceq f$ and hence we have that ${(P\circ f)(x,y) \subseteq c(x,y)}$ for all $c(x,y)\in C$.
So, if there is an assignment satisfying $P\circ f$, the same assignment must also satisfy $I$.
\end{proof}

We are now ready to present the limited set of partial orders that we are interested in enumerating.

\begin{definition}
A \emph{total ordering of pairs} ({\toop}) is a partial order $P=(V,\leq_P)$ such that for every $x\in V$ there is at most one $y\in V\setminus \{x\}$ such that $x||_P y$.
We say that $(x,y)$ is a \emph{pair} in $P$ if $x||_P y$. 
\end{definition}

We will use the notation $(a_1,b_1)<_P\ldots<_P(a_i,b_i)$ for writing (part of) a {\toop} $P$.
Here $(a_j,b_j)$ are our pairs for $j\in\{1,\ldots,i\}$.
Note that we do not make a difference between $(a_j,b_j)$ and $(b_j,a_j)$.

Now we are ready to start applying {\toop}s on multi relational networks, see what important structures can occur and, most importantly, investigate why this is beneficial and determine for which cases the output is still difficult to solve.

\begin{definition}
Given a {\toop} $P$ and a multi relational network $f$ of a {\pot} instance, a \emph{link} in $P\circ f$ is a non-empty sequence of pairs $(a_1,b_1)<_P\ldots<_P(a_n,b_n)$ such that
\begin{enumerate}
    \item $(P\circ f)(a_i,b_i)\setminus\{=\}=\{<,>\}$,
    \item $(P\circ f)(a_i,a_j)\setminus\{=\}=\{<\}$, $1\leq i<j\leq n$,
    \item $(P\circ f)(b_i,b_j)\setminus\{=\}=\{<\}$, $1\leq i<j\leq n$, and
    \item $|| \in (P\circ f)(a_i,b_j)$, $i\neq j$.
\end{enumerate}
Two links \emph{overlap} if they share any pairs.
\end{definition}
Note that a single link does not need to contain more than a single pair.
Also, two links sharing the same pairs, e.g. $(a,b)$ and $(b,a)$, are technically the same link.
This is just two different representations, or \emph{directions}, of the same link, which will be relevant for the following important concept. 

\begin{definition}
Given a {\toop} $P$ and multi relational network $f$ for some {\pot} instance, a \emph{chain} in $P \circ f$ is two variables $x<_Py$ with ${(P\circ f)(x,y)\setminus\{=\}=\{||\}}$ and a non-empty sequence of $m$ links
\begin{align*}
(a_{1,1},b_{1,1})&<_P\ldots<_P(a_{1,n_1},b_{1,n_1})<_P\ldots<_P\\
(a_{m,1},b_{m,1})&<_P\ldots<_P(a_{m,n_m},b_{m,n_m})
\end{align*}
such that
\begin{enumerate}
    \item either $x=b_{1,1}$ or $(P\circ f)(x,b_{1,1}) \subseteq \{<,=\}$,
    \item either $y=a_{m,n'}$ or $(P\circ f)(a_{m,n'},y) \subseteq \{<,=\}$,
    \item $P\circ f(a_{i,n_i},b_{i+1,1})\setminus\{=\}=\{<\}$,
    \item $(P\circ f)(a_{i,i'},a_{j,j'})\setminus\{=\}=\{<,||\}$, $i < j$, ${i'\in\{1,\dots,n_i\}}$, ${j'\in\{1,\dots,n_j\}}$ and
    \item $(P\circ f)(b_{i,i'},b_{j,j'})\setminus\{=\}=\{<,||\}$, $i< j$, ${i'\in\{1,\dots,n_i\}}$, ${j'\in\{1,\dots,n_j\}}$.
\end{enumerate}
We say that a chain is {\em broken} if ${(P\circ f)(a_{i,i'},b_{i,i'})=\{<\}}$ for all $i'$ for some $i$.
The \emph{length} of a chain is the number of links $m$ in the chain.
We say that $x$ is the \emph{head} of the chain and $y$ the \emph{tail}.
\end{definition}

While we here differ between $(a,b)$ and $(b,a)$ in contrast to how we did for {\toop}s, it should be clear by the context when it is of importance or not, i.e. when we are speaking of chains and when we are only discussing (parts of) a {\toop}.

\begin{example}\label{ex:ChainsAndLinks}

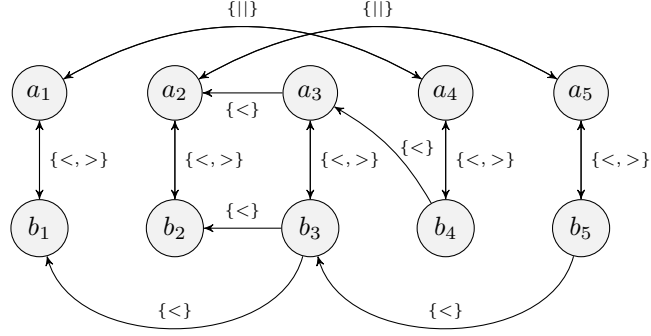
\begin{figure}
  \begin{tikzpicture}[scale=0.9]

    \node[circle,draw, fill = gray!10] (a1) at (0,0) {$a_1$};
    \node[circle,draw, fill = gray!10] (a2) at (2,0) {$a_2$};
    \node[circle,draw, fill = gray!10] (a3) at (4,0) {$a_3$};
    \node[circle,draw, fill = gray!10] (a4) at (6,0) {$a_4$};
    \node[circle,draw, fill = gray!10] (a5) at (8,0) {$a_5$};
    
    \node[circle,draw, fill = gray!10] (b1) at (0,-2) {$b_1$};
    \node[circle,draw, fill = gray!10] (b2) at (2,-2) {$b_2$};
    \node[circle,draw, fill = gray!10] (b3) at (4,-2) {$b_3$};
    \node[circle,draw, fill = gray!10] (b4) at (6,-2) {$b_4$};
    \node[circle,draw, fill = gray!10] (b5) at (8,-2) {$b_5$};
    
    \draw[->] (a1) -- (b1) node[midway, right] {\tiny $\{<,>\}$};
    \draw[->] (a2) -- (b2) node[midway, right] {\tiny $\{<,>\}$};
    \draw[->] (a3) -- (b3) node[midway, right] {\tiny $\{<,>\}$};
    \draw[->] (a4) -- (b4) node[midway, right] {\tiny $\{<,>\}$};
    \draw[->] (a5) -- (b5) node[midway, right] {\tiny $\{<,>\}$};

    \draw[->] (a3) -- (a2) node[midway, below] {\tiny $\{<\}$};
     \draw[->] (b3) -- (b2) node[midway, above] {\tiny $\{<\}$};

    \draw[->] (b4)  edge[bend right = 15] node [right] {\tiny $\{<\}$} (a3);

    \path[->] (a1) edge[bend left] node [above] {\tiny $\{||\}$} (a4);
    \path[->] (a4) edge[bend right] node [above] {} (a1);

    \path[->] (a2) edge[bend left] node [above] {\tiny $\{||\}$} (a5);
    \path[->] (a5) edge[bend right] node [above] {} (a2);

    \path[->] (b3) edge[bend left = 75] node [above] {\tiny $\{<\}$} (b1);
    \path[->] (b5) edge[bend left = 75] node [above] {\tiny $\{<\}$} (b3);

    \draw[->] (b1) -- (a1); 

    \draw[->] (b2) -- (a2);
    \draw[->] (b3) -- (a3);
    \draw[->] (b4) -- (a4);
    \draw[->] (b5) -- (a5);
    
    \end{tikzpicture}
  \caption{Graphic representation of two chains overlapping in one link ($a_3,b_3$), as described in Example~\ref{ex:ChainsAndLinks}.
  Variables/pairs to the left are ordered before those to the right in our {\toop} and hence we assume either a $<$- or a $||$-relation going from left to right.}
  \label{fig:example1}
\end{figure}

Take the {\toop} \\
${(a_1,b_1)<_P(a_2,b_2)<_P(a_3,b_3)<_P(a_4,b_4)<_P(a_5,b_5)}.$
Let ${f(b_1,b_3)=\{<\}}$,
${f(a_2,a_3)=\{<\}}$, ${f(b_2,b_3)=\{<\}}$,
${f(a_3,b_4)=\{<\}}$,
${f(b_3,b_5)=\{<\}}$, ${f(a_1,a_4)=\{||\}}$, ${f(a_2,a_5)=\{||\}}$ and for every ${i\in\{1,\ldots,5\}}$,  ${f(a_i,b_i)=\{<,>\}}$.
We now have five links: $(a_1,b_1)$, ${(a_2,b_2)<_P(a_3,b_3)}$, $(b_3,a_3)$, $(b_4,a_4)$ and $(b_5,a_5)$.
We also have two chains: ${(a_1,b_1)<_P(b_3,a_3)<_P(b_4,a_4)}$ and ${(a_2,b_2)<_P(a_3,b_3)<_P(b_5,a_5)}$.
Here ${(a_2,b_2)<_P(a_3,b_3)}$ and $(b_3,a_3)$ contain the same pair in $a_3,b_3$, but their directions differ, and hence these links overlap in the opposite directions.

This example is also visualized in Figure~\ref{fig:example1}.
\end{example}

Before we prove that $T\circ f$ is solvable in polynomial time for any total order $T$ we introduce the following comparability property of partial orders. 

\begin{definition}\label{def:||larger}
Given two relational networks $f$ and $g$ over the same variable set $V$, $f$ is \emph{$||$-larger} than $g$ if, for any variables $x,y \in V$:
\begin{enumerate}
    \item if $g(x,y)=\{||\}$ then $f(x,y)=\{||\}$,
    \item if $g(x,y)=\,\{<\}$ then $f(x,y)\in\{<,||\}$,
    \item if $g(x,y)=\,\{>\}$ then $f(x,y)\in\{>,||\}$,
    \item if $g(x,y)=\{=\}$ then $f(x,y)\in\{<,>,||,=\}$, 
\end{enumerate}
and there exists $x,y \in V$ such that $g(x,y)\neq f(x,y)$.

Furthermore, we say that $f$ is \emph{$||$-maximal} if there does not exist any $g$ which is $||$-larger than $f$.
\end{definition}

\subsection{The algorithm}

We begin by showing how to solve ${T\circ f}$ in polynomial time for a total order $T$.

\begin{lemma}\label{lem:todotfpoly}
Let $I = (V,C)$ be a {\pot} instance with multi relational network $f$ and let $T = (V, <_T)$ be a total order.
Then ${T\circ f}$ is solvable in polynomial time.
\end{lemma}
\begin{proof}
Solving ${T\circ f}$ can be done by enforcing consistency for each triple ${x,y,z\in V}$, e.g., if ${f(x,y)=\{<\}}$ and ${f(y,z)=\{<\}}$ then ${f(x,z)=\{<\}}$.
We repeat this until no more changes occur.
If any ${f(x,y)=\emptyset}$ then ${T\circ f}$ must be a 'no'-instance, otherwise it is a 'yes'-instance.
Soundness for this approach follows naturally.
To see that this approach is complete take the $||$-maximal relational network $g$ which can give $T$ when topologically sorted.
Since $I$ given $T$ is a 'yes'-instance, $g$ must exist.
After the first round of local consistency propagation the following must be true for all triples $x<_T y <_T z$ for our $f$:
\begin{enumerate}
    \item if ${=} \in f(x,z)$ then ${=} \in f(x,y)$ and ${=} \in f(y,z)$, or ${||} \in f(x,y)$ and ${||\,\in f(y,z)}$,
    \item if ${<} \in f(x,z)$ then $f(x,y)= f(y,z)$ cannot equal ${=}$,  and
    \item if ${||} \in f(x,z)$ then ${||} \in f(x,y)$ or ${||} = f(y,z)$.
\end{enumerate}
Similarly, $g(x,y)\in f(x,y)$.
By then removing $=$-relations from any $|f(x,y)|>1$ the relations $g(x,y)\in f(x,y)$ must still be true and hence the network is still a 'yes'-instance.
Again, we enforce consistency for triples until no changes occur.
For any triple ${x<_T y <_T z}$, where $g(x,y)$, $g(x,z)$ and $g(y,z)$ are not equal to $\{=\}$, it then still holds that if ${||\in f(x,z)}$ then ${||\in f(x,y)}$ or ${||\in f(y,z)}$.
Hence, for every ${||\in f(x,y)}$ we can assume ${f(x,y)=\{||\}}$ and still have a 'yes'-instance. 
In fact, $f=g$ since $g$ is $||$-maximal.
Hence, if $I$ is a 'yes'-instance and if there is any relational network for $I$ that topologically sorts to $T$, we will answer 'yes'.

For the complexity, we only need to propagate complexity for variable triples, and each propagation cycle must remove at least one relation from $f$, else no change is made.
As there are at most four possible atomic relations between each variable pair, the number of cycles is polynomially bounded by $O(|V|^2)$.
Hence, this approach can be done in polynomial time.
\end{proof}

Why it is actually interesting to solve $T\circ f$ in polynomial time is shown by the following lemma.

\begin{lemma}\label{lem:IyesTyes}
If a {\pot} instance $I$ with multi relational network $f$ is a 'yes'-instance then there is a total order $T$ such that $T\circ f$ is a 'yes'-instance.
\end{lemma}
\begin{proof}
We remind the reader that we can represent a solution for $I$ with a relation network.
Take any relational network $g$ satisfying $I$ and the partial ordering $P$ described by this relational network.
Topologically sort $P$ to a total order $T$.
From the definition of $T\circ f$ we know that ${T\circ f(x,y) = f(x,y)\setminus\{>\}}$ if $x<_Ty$.
As $x\not <_Ty$ if $g(x,y)=\{>\}$, $g$ must be consistent with $T\circ f$ and hence $g$ is also satisfies $T\circ f$, proving the lemma.
\end{proof}

As a sanity check for making sure that an instance does not have small subinstances that are 'no'-instanses, i.e. the instance is localy consistent, we prove the following lemma.
Recall that $||f||$ is the number of bits needed to represent $f$.

\begin{lemma} \label{lemma:local}
For any multi relational network $f$, integer $k>0$ and a set of variables $s\subseteq V$ with $|s|\leq k$, a maximally general multi relational network $f' \preceq f$ which is locally consistent with $s$ can be computed in $h(k)\cdot ||f||^{O(1)}$ time for some computable function $h(k)$.
\end{lemma}
\begin{proof}
Let  ${\{g_1,\dots,g_m\}}$ be the set of all consistent relational networks $g\preceq f$ over $s$.
For each $s$, there are roughly $16^{|s|^2}$ potential $g_i$s as we have four different relations that can either be allowed, or not allowed, between each pair of variables in $s$.
Each potential $g_i$ can be compared to $f$ to check for consistency in $||f||^{O(1)}$ time as $||g_i|| \leq ||f||$.
Now, for each pair ${x,y\in s}$ let ${f'(x,y) = \bigcup g_i(x,y)}$ and for all pairs where either ${x\not \in s}$ or ${y\not \in s}$ let ${f'(x,y)=f(x,y)}$.
Since ${|s|\leq k}$ the number of possible $16^{|s|^2}$ and the complexity of constructing $f'$ is bounded by $h(k)$ for some computable function ${h \colon \mathbb{N} \rightarrow \mathbb{N}}$.
Clearly, $f'\preceq f$, and since we enumerated all possible relations between all pairs in $s$, $f'$ must also be the largest such multi relational network.
\end{proof}

Now, we are ready to introduce the first of the two main gadgets necessary for our later proofs. For a partial order $P = (V, \leq)$ and distinct $x,y \in V$ we write $P^{x < y}$ for ${(V, \leq \cup \{(x,y)\})}$. 
While this operation may technically produce structures that are not partial orders (e.g., ${x <_P y <_P x}$), we in the forthcoming definitions will only use it on pairs in {\toop}s where the result is always guaranteed to be a partial order.

\begin{definition}\label{def:R_partial}
For a {\toop} $P = (V,\leq)$ and a multi relational network $f$ over $V$ we describe a {\em partial} function $\rpartial$ defined according to the following rules.

\begin{enumerate}
\item
$\rpartial(P,f) = (P,f)$ if $P$ is a total order,
\item
$\rpartial(P,f) = R(P^{x < y},f)$ if there exists distinct ${x,y\in V}$ with $x||_Py$ and $< \not\in (P\circ f)(x,y)$,
\item
$\rpartial(P,f) = R(P^{a_i < b_i})$ if there exists a link ${(a_1,b_1)<\dots<(a_k,b_k)}$ such that there is no other link ${(c_1,d_1)<\dots<(c_{k'},d_{k'})}$ for which ${a_i=d_{i'}}$ and ${b_i=c_{i'}}$ for any $i$ and $i'$, 
\item
$\rpartial(P,f) = \rpartial(P, f')$  
 if there exists $a,b,c,d \in V$ and a maximally general multi relational network $f' \neq f$, $f' \preceq f$ and locally consistent with $\{a,b,c,d\}$.
\end{enumerate}
\end{definition}

It is easy to see that if $\rpartial(P,f) = (T,f)$ is defined then it returns a total order $T$ such that $P$ is a \stub{} of $T$. We then consider the following extension of $\rpartial$ which is guaranteed to be totally defined due to the second rule.

\begin{definition}\label{def:R_total}
For a {\toop} $P = (V, \leq)$ and a multi relational network $f$ over $V$ we describe a {\em total} function $\rtotal$ defined according to the following rules.

\begin{enumerate}
\item
$\rtotal(P,f) = \rpartial(P,f)$ if $\rpartial(P,f)$ is defined, and
\item
$\rtotal(P,f) = \rtotal(P^{x < y}, f)$ if there exists $x,y\in V$ where $\{<,>\}\subseteq (P\circ f)(x,y)$.
\end{enumerate}
\end{definition}

Note that $\rtotal$ will always return a total order, since otherwise the second rule could be applied.
However, the function is not sound and can given a 'yes'-instance $P\circ f$ return a 'no'-instance $T\circ f$.
Even so, this function will be the one we later use for solving {\pot} instances.

As a second extension of $\rpartial$, we introduce the total function $\rcorrect$.
This function will computationally be more expensive than $\rtotal$, but with the trade-off that it in the second step makes sure extending the given {\toop} with $x < y$ is actually a reasonable choice that does not get the function stuck with a 'no'-instance.
Thus, if extending the partial order by $x < y$ would yield a 'no'-instance, $\rcorrect$ (in contrast to $\rtotal$) simply stops and returns $(P,f)$.

\begin{definition}\label{def:R_correct}
For a {\toop} $P = (V, \leq)$ and a multi relational network $f$ over $V$ we describe a total function $\rcorrect$ defined according to the following rules.

\begin{enumerate}
\item
$\rcorrect(P,f) = \rpartial(P,f)$ if $\rpartial(P,f)$ is defined.
\item
$\rcorrect(P,f) = \rcorrect(P^{x < y}, f)$ if there exists ${x,y\in V}$ with $\{<,>\}\subseteq (P\circ f)(x,y)$ and where  ${P\circ (f\cup \{x\{<,||,=\}y\})}$ is a 'yes'-instances, and
\item  $\rcorrect((P,f)) = (P,f)$ otherwise.
\end{enumerate}
\end{definition}

Note that $\rcorrect$ differs from $\rtotal$ in the sense that if $\rcorrect$ (via the second rule) is given a 'yes'-instance it also guarantees that the output is a 'yes'-instance.
The precise relation between $\rtotal$ and $\rcorrect$ will be of much interest to us.
In fact, we will show that if $f$ is a 'yes'-instance, then there must exist a {\toop} such that $\rtotal=\rcorrect$. We begin with the following lemma.

\begin{lemma}\label{lem:Rpoly}
For an arbitrary {\pot} instance with multi relational network $f$ and some arbitrary {\toop} ${P=(V,\leq_P)}$, ${\rtotal(P,f)}$ can be computed in ${poly(||I||)}$ time and space.
\end{lemma}
\begin{proof}
Step 1 and 2 of $\rpartial$ and step 2 of $\rtotal$ are quite clearly polynomial.
Since there are $|V|^4$ sets of four variables, step 4 of $\rpartial$ is polynomial by Lemma~\ref{lemma:local}.
Last, step 3 of $\rpartial$ is polynomial since there are at most $|V|/2$ pairs, and finding if they are part of some chain (and in which directions) can be done in polynomial time.
Hence, all steps of $\rtotal$ are doable in polynomial time.
\end{proof}

Using $\rtotal(P,f)=(T,f')$ we are now left with ${T\circ f'}$, which we can solve in polynomial time according to Lemma~\ref{lem:todotfpoly}, and if this approach returns $\mathit{true}$ then $I$ must be a 'yes'-instance (via Lemma~\ref{lem:PfyesIyes}).

Returning to the relationship between $\rtotal$ and $\rcorrect$ we are now ready to describe the property needed for ${\rtotal(P,f)\neq \rcorrect(P,f)}$ to occur.

\begin{lemma}\label{lem:RRoracleDiff}
Let $P = (V,\leq_P)$ be a {\toop} and $f$ a multi relational network over $V$ such that  $P\circ f$ is a 'yes'-instance. Then 
$\rtotal(P,f)\neq \rcorrect(P,f)$ only if $P\circ f$ contains a non-empty set $S$ of chains of length at least two and such that every link of every chain in the set overlaps with some other link in the opposite direction. 

\end{lemma}
\begin{proof}
Assume $\rtotal\neq \rcorrect$ and $S=\emptyset$.
The only step where $\rtotal$ and $\rcorrect$ differ in output is in step 2. 
Further, assume we have a topological sorting $T$ of some solution for $P\circ f$.
Whenever $\rtotal$ reaches step 2 and for a pair $x||_Py$ chooses $x\{<,=,||\}y$, without knowing if this is a good choice or not, we have four different cases that could theoretically occur and are worth considering.

\begin{enumerate}
    \item If $x=y$ or $x||y$ in our solution for $P\circ f$, then $x<_Ty$ and $y<_Tx$ are both valid and will yield the same solution.
    \item If $x<y$ in our solution, but $y<x$ in some other, but otherwise identical solution, then $x<_Ty$ and $y<_Tx$ are both valid and yield 'yes'-instances.  
    \item If $x<y$ in our solution and there is no solution with $y<x$ but which is otherwise identical, then there must be two variables $u,v$ such that $x\leq u$, $v\leq y$ and $u||v$.
    This describes a link and a chain.
    Then, either this chain has length one, in which point Step 4 of $\rpartial$ will be applicable as local consistency over $\{x,y,u,v\}$ would not allow the case of $y<x$,  
    or this chain has length longer than two, but contains a link not overlapping with any other link in the opposite direction. 
    Since there is here a link that does not overlap with any other in the opposite direction, step 3 of $\rpartial$ would be applicable and we would set $x<y$ or $y<x$ depending on which one breaks the chain.
    \item There are chains of length two or more, but no chain contains any link that does not overlap with some other link in the opposite direction.
    This matches the definition of chains in $S$, and hence $S$ is non-empty.
\end{enumerate}

In the first two of these cases $\rtotal$ and $\rcorrect$ behave identically, and hence yield the same output.
For the third one, neither $\rtotal$ nor $\rcorrect$ reaches their respective step 2, and hence they behave identically. 
For the fourth case, $S$ is non-empty.
So neither of these four cases satisfies our assumption.
But our four cases are exhaustive: they cover all relations between $x$ and $y$ in solutions for $T\circ f$ and all cases for all these relations.
So the initial assumptions must be false, meaning that either $\rtotal=\rcorrect$ or $S\neq \emptyset$, completing the proof.
\end{proof}

Before we introduce the final piece of the puzzle, we give two short definitions: one for a sub-class of {\toop}s that are easier to enumerate and one for our notation for ${\rtotal(P,f)=\rcorrect(P,f)}$.

\begin{definition}
A \emph{proper total ordering of pairs} ({\ptoop}) is a {\toop} $P=(V,\leq_P)$ such that for every $x\in V$ either there is exactly one $y\in V\setminus \{x\}$ such that $x||_P y$ or $y<_Px$ for all $y\in V\setminus \{x\}$.
\end{definition}

\begin{definition}
If for some arbitrary {\toop} $P=(V,\leq_P)$ and relational network $f$ over $V$, ${\rtotal(P,f)=\rcorrect(P,f)}$ we say that $P$ is \emph{reducible}
for $f$.
\end{definition}

And now, finally, we can show, and prove, why $\rtotal$, $\rcorrect$, reducibility and (P){\toop}s are interesting for multi relational networks for {\pot} instances.

\begin{lemma}\label{lem:binaryPartialOrderExistance}
For every 'yes'-instance $I$ of {\pot} with multi relational network $f$, there is a {\ptoop} reducible for $f$.
\end{lemma}
\begin{proof}
Since $I$ is a 'yes'-instance we have a relational network $g$ satisfying $I$, and by Lemma~\ref{lem:IyesTyes} we know there is a total order $T$ such that $T\circ f$ is also satisfied by $g$.

We start the process of finding a {\toop} reducible for $f$ by generating (from $T$) an arbitrary {\ptoop} $P$ such that $P$ is an \stub{} of $T$. If $P$ is reducible for $f$, then return $P$.
Otherwise if $P$ is not reducible for $f$ we assume that ${\rcorrect(P,f)=(P,f)}$.
Since $\rcorrect$ can remove chains of length one by local consistency over quads, we know by Lemma~\ref{lem:RRoracleDiff} that $P\circ f$ contains chains longer than two, and such that their links overlap with others in the opposite direction.
I.e., we have a link \[{(a_{1,1},b_{1,1})<_P\ldots<_P(a_{1,n},b_{1,n})}\] and a second link  \[{(u_{1,1},v_{1,1})<_P\ldots<_P(u_{1,m},v_{1,m})}\]
such that there is some $a_{1,i}=v_{2,i'}$ and $b_{1,i}=u_{2,i'}$.
Furthermore, we also have a link \[{(a_{2,1},b_{2,1})<_P\ldots<_P(a_{2,n'},b_{2,n'})}\] that is part of the same chain as \[{(a_{1,1},b_{1,1})<_P\ldots<_P(a_{1,n},b_{1,n})}\] and some link 
\[{(u_{2,1},v_{2,1})<_P\ldots<_P(u_{2,m'},v_{2,m'})}\] that is part of the same chain as \[{(u_{1,1},v_{1,1})<_P\ldots<_P(u_{1,m},v_{1,m})}.\]
From all such overlaps, take the overlapping pair with the lowest index in $P$ and call it $(x,y)$.
Without loss of generality, assume that $g(x,y)= {<}$ in our relational network $g$ satisfying both $I$ and $T\circ f$.
Take the variable ${y<_Tz}$ such that $g(x,z)=g(z,y)=||$, and such that there is no other variable $z'<_Tz$ fulfilling the same conditions. 
To see that such a variable $z$ actually exists, take the tail $x'$ and $y'$ of the chains overlapping in opposite directions in $(x,y)$.
Since the chain which includes $x$ and $x'$ must be broken, and similarly for $y$ and $y'$, we must have that $g(x,x')=||$ and $g(y,y')=||$.
Assume that $g(y,x')=<$.
But since ${g(x,y)=g(y,x')=<}$ then $g(x,x')=<$ while we assumed $g(x,x')=||$, so we have a contradiction.
Hence, we have ${g(x,x')=g(x',y)=||}$ and so at least one variable with this property exists and so does our $z$.
In $T$ there is now a sequence ${x<_Ty<_Tx_1<_T\ldots<_Tx_{j}<_Tz}$.
Since $z$ is incomparable to both $x$ and $y$ in $g$ and as $z$ is the least indexed such variable in this sequence, then for all $x_i$ in the sequence we have $g(z,x_i)=||$.
Hence, if we construct a new total order $T'$ such that ${x<_{T'}z<_{T'}y<_{T'}x_1<_{T'}\ldots<_{T'}x_{j}}$ but which is otherwise identical to $T$, then $g$ also satisfies $T'\circ f$.
If we now also construct a new arbitrary {\toop} $P'$ from $T'$ our two original links are no longer links for any chains.
We repeat this process until the resulting {\toop} $P'$ is reducible for $f$.

The question is now if the above construction of a reducible {\ptoop} for $f$ halts.
To answer this we make the following observation:
for each recursion we choose some index $i$ in $P$ such that the pair at $i$ in $P\circ f$ is changed from a link to a non-link in $P'\circ f$, without making any new links out of pairs of index less than $i$ that have previously been chosen.
Hence the number of pairs at index $i$ or less that cannot become links by choosing a new index $j>i$ has increased by one.
Call the pairs at these indexes selected \emph{protected}.
Note that a pair at index $i$ loses its protected status and can become a link again if $j<i$ is chosen, but then the same logic applies to $j$ instead: the number of protected pairs at index $j$ or less has increased.
So, in each iteration, either $i$ decreases or the number of protected pairs indexes less than $i$ increases.
Hence, the function must reach a point where no pair can be selected and $P$ is reducible for $f$.
\end{proof}

One may ask what happens to chains in $P\circ f$ that are satisfied by equality in the solution and are hence technically never broken.
For all of the links in these chains, the local consistency check in Step 4 of $\rpartial$ will keep the equality relation, while Step 2 in $\rtotal$ will always produce a 'yes'-instance given a 'yes'-instance when working on the pairs in these links.
Hence, these are a non-issue and will be satisfied by our greedy approach
(as long the variables are equal in some solution).

We now have everything we need to present the main result for this section:

\begin{theorem}\label{thm:potfinal}
Any arbitrary {\pot} instance $I=(V,C)$ with $|V|=n$ is solvable in $\Ordo^*(n!/2^\frac{n}{2}) \subseteq \Ordo^*((0.2601n)^n)$ time and $poly(||I||)$ space.
\end{theorem}
\begin{proof}
By Lemma~\ref{lem:PfyesIyes} we know that if we find any partial order $P$ such that $P\circ f$ is a 'yes'-instance, then $I$ must be a 'yes'-instance.
From Lemma~\ref{lem:binaryPartialOrderExistance} we know there exists a {\ptoop} $P$ such that $P\circ f$ is a 'yes'-instance if $I$ is and such that we can solve $P\circ f$ in polynomial time by applying Lemma~\ref{lem:todotfpoly} to the output of $\rtotal(P,f)$, and by Lemma~\ref{lem:Rpoly} $\rtotal(P,f)$ can be calculated in polynomial time.
There are $n!/2^\frac{n}{2}$ such {\ptoop}s, and we can, in a polynomial factor on $n$, enumerate over all {\ptoop}s over $V$ until we find our $P$.
If no 'yes'-instance $P\circ f$ is found in this way then $I$ must be a 'no'-instance, again by Lemma~\ref{lem:binaryPartialOrderExistance}, and we can safely answer 'no'.
Via Stirling's approximation of $n!$ we obtain $\Ordo^*((0.2601n)^n)$.
\end{proof}

With Theorem~\ref{thm:potfinal} in place, we have now seen the entire chain of reasoning of how enumerating {\ptoop}s is enough to solve any {\pot} instance. While providing a new state of the art.

\section{Conclusion and Discussion}\label{sec:conc}

In this paper we used structural properties on potential solutions for {\pot} instances to achieve a new and significantly improved upper bound.
These results, and the techniques used, raise new questions.

A major open question is whether our algorithm for {\pot } can be adapted to other qualitative reasoning tasks such as RCC-5 or RCC-8.
The issue here seems to be that the constraints \emph{partially overlapping} and \emph{disjoint} are harder to handle than just incomparability.
Combined, however, these relations behaves identical to incomparability, so the idea is by no means far-fetched.
Partially overlapping enforces a form of upwards transitivity in that every set containing a set $A$ must also be at least partially overlapping with some set $B$ if $A$ and $B$ are partially overlapping.
Similarly, disjoint enforces a form of downwards transitivity in that every set contained in $C$ must be disjoint to $D$ if $C$ is disjoint to $D$.
Effectively this introduces some different forms of chains than for {\pot} that require some novel ideas to handle.
As a first step, finding a method to handle these chains faster than $2^n$ for $n$ links would be enough to improve the current state-of-the-art.

It is also natural to ask if it is possible to avoid enumerating certain orderings of pairs and thus push down the runtime even further.
For example, can we find certain orderings of pairs that will never satisfy our instances?
And can we do this fast enough, and often enough, that it yields a significant improvement to the overall runtime?
Another promising idea is to  partition variables into triples instead of pairs, or quads
that are later further restricted to pairs in an intelligent manner.
This requires non-trivial changes to our algorithm but does not seem impossible.
If this approach can be pushed further, it would likely hit a limit of $\Ordo^*(\sqrt{cn}^n)$ where the local consistency checks would start bottlenecking.
Even so, this would be the first $o(n)^n$ complexity result for {\pot} and a major advance in solving NP-hard qualitative reasoning problems.

\section*{Acknowledgements}
The second author is partially supported by the Swedish research council (VR) under grant 2019-03690.

\bibliographystyle{abbrv}
\bibliography{bib}

\end{document}